\documentclass{article}
\usepackage{ijcai25}

\usepackage{times}
\usepackage{soul}
\usepackage{url}
\usepackage[hidelinks]{hyperref}
\usepackage[utf8]{inputenc}
\usepackage[small]{caption}
\usepackage{graphicx}
\usepackage{amsmath}
\usepackage{amsthm}
\usepackage{amssymb}
\usepackage{amsfonts}
\usepackage{booktabs}
\usepackage{algorithm}
\usepackage[noend]{algpseudocode}
\usepackage[switch]{lineno}

\usepackage{natbib}
\usepackage{cleveref}
\usepackage{bm}

\newtheorem{theorem}{Theorem}[section]

\newtheorem{lemma}[theorem]{Lemma}

\theoremstyle{definition}
\newtheorem{definition}[theorem]{Definition}

\newcommand{\cG}{\mathcal{G}}

\newcommand{\cZ}{\mathcal{Z}}

\newcommand{\tnu}{\tilde{\nu}}

\newcommand{\R}{\mathbb{R}}
\newcommand{\E}{\mathbb{E}}
\newcommand{\eps}{\epsilon}
\newcommand{\cD}{\mathcal{D}}

\newcommand{\var}{\mathrm{Var}}

\newcommand{\argmin}{\mathrm{argmin}}
\newcommand{\fav}{\mathrm{fav}}
\newcommand{\ind}{\mathbf{1}}
\newcommand{\algdiv}{\mathrm{ALG}_{\mathrm{div}}}
\newcommand{\algpropsmall}{\mathrm{ALG}^{\mathrm{PROP\text{-}small}}}

\algnotext{EndFor}
\algnotext{EndIf}
\algnotext{EndWhile}
\algnotext{EndProcedure}

\newcommand*{\diff}[1]{\mathop{}\!\mathrm{d}#1} %% https://tex.stackexchange.com/questions/60545/should-i-mathrm-the-d-in-my-integrals

\title{Asymptotic Fair Division: Chores Are Easier Than Goods}

\author{
Pasin Manurangsi$^1$\And
Warut Suksompong$^2$\\
\affiliations
$^1$Google Research, Thailand\\
$^2$National University of Singapore, Singapore\\
}

\allowdisplaybreaks

\begin{document}

\maketitle

\begin{abstract}
When dividing items among agents, two of the most widely studied fairness notions are envy-freeness and proportionality. 
We consider a setting where $m$ chores are allocated to $n$ agents and the disutility of each chore for each agent is drawn from a probability distribution.
We show that an envy-free allocation exists with high probability provided that $m \ge 2n$, and moreover, $m$ must be at least $n+\Theta(n)$ in order for the existence to hold.
On the other hand, we prove that a proportional allocation is likely to exist as long as $m = \omega(1)$, and this threshold is asymptotically tight.
Our results reveal a clear contrast with the allocation of goods, where a larger number of items is necessary to ensure existence for both notions.
\end{abstract}

\section{Introduction}
\label{sec:intro}

In applications of dividing items among agents, such as distributing teaching loads to faculty members or splitting household tasks between partners, fairness is often a key consideration \citep{BramsTa96,Moulin03}.
To formally reason about fairness, one must define what it means to be ``fair''.
Among the plethora of fairness notions proposed in the literature, two of the most widely used are envy-freeness and proportionality.
An allocation is called \emph{envy-free} if no agent prefers to swap her bundle of items with that of another agent, and \emph{proportional} if every agent receives no worse than her proportionally fair share, defined as $1/n$ of her value for the entire set of items with $n$ denoting the number of agents.

When the items to be allocated are \emph{indivisible}---that is, each item must be allocated in its entirety to one agent---an envy-free or proportional allocation may not exist.
For example, when $m$ valuable \emph{goods} are divided among $n > m$ agents, at least one agent will be left empty-handed; such an agent will be envious and deprived of her proportional share.
In light of this, \citet{DickersonGoKa14} asked the following natural question: if the agents' utilities for goods are drawn at random from a distribution, for which asymptotic relations between $m$ and $n$ does an envy-free allocation exist with high probability (meaning that the probability of existence approaches $1$ as $n\rightarrow\infty$)?
They showed that the asymptotic existence occurs when $m = \Omega(n\log n)$, but not when $m = n + o(n)$.
\citet{ManurangsiSu20,ManurangsiSu21} closed this gap by proving that if $m$ is divisible by $n$, an envy-free allocation is already likely to exist as long as $m\ge 2n$, whereas if $m$ is not ``almost divisible''\footnote{This means that the remainder upon dividing $m$ by $n$ is neither less than $n^{o(1)}$ nor greater than $n - n^{o(1)}$.} by $n$, existence is guaranteed only when $m = \Omega(n\log n/\log\log n)$.
For proportionality, \citet{ManurangsiSu21} showed that existence is likely whenever $m\ge n$; this is tight since no proportional allocation exists when $m < n$.

In this paper, we investigate asymptotic fair division of \emph{chores}, i.e., items that yield disutility to agents.
Due to its broad applicability---including the division of teaching loads, medical shifts, and household tasks---chore allocation has received significant attention in the fairness literature \citep{GuoLiDe23}.
As we shall demonstrate, the asymptotic results for chores are markedly different from those for goods, with respect to both envy-freeness and proportionality.

\subsection{Our Contributions}

In our model, each agent's disutility for each chore is drawn independently from a non-atomic distribution $\cD$ supported on $[0,1]$, and the agents' disutilities are additive across chores.
Following prior work in this domain, we assume for most of our results that $\cD$ is \emph{PDF-bounded}, meaning that its probability density function is bounded below and above by some positive constants $\alpha$ and~$\beta$, respectively \citep{ManurangsiSu21,BaiGo22,YokoyamaIg25}.
We describe the model formally in \Cref{sec:prelim}.

Our first result is that an envy-free allocation of the chores exists (and can be computed efficiently) with high probability as long as $m \geq 2n$.

\begin{theorem} \label{thm:upper-main}
For any PDF-bounded distribution $\cD$, if $m \geq 2n$, then with high probability, an envy-free allocation exists.
Moreover, there is a polynomial-time algorithm that computes such an allocation.
\end{theorem}

\Cref{thm:upper-main} stands in contrast to the case of goods, where if $m$ is not almost divisible by $n$, then $m$ needs to be at least $\Omega(n\log n/\log\log n)$ for an envy-free allocation to exist with high probability \citep{ManurangsiSu20}.
On the other hand, if $m < n$, clearly no envy-free allocation of the chores exists.
We strengthen this observation by showing that when $m$ exceeds $n$ by a small (multiplicative) constant factor, an envy-free allocation is still unlikely to exist.
Note that this result does not require the PDF-boundedness assumption.

\begin{theorem} \label{thm:lb-main}
Let $\nu \geq 1.1256$ be the (unique) positive solution to the equation
\begin{align*} 
\nu\left(1 + \left(1 + \frac{1}{\nu}\right)e^{-\frac{1}{\nu}}\right) = 2.
\end{align*}
For any constant $\eps \in (0, \nu)$, if $m \leq (\nu - \eps)n$, then with high probability, no envy-free allocation exists.
\end{theorem}

On the proportionality front, we establish that a surprisingly small number of chores---any superconstant number---is already sufficient to ensure existence.

\begin{theorem} \label{thm:prop-ub-main}
For any PDF-bounded distribution $\cD$, if $m = \omega(1)$, then with high probability, a proportional allocation exists.
Moreover, there is a polynomial-time algorithm that computes such an allocation.
\end{theorem}

We also provide a matching bound, which shows that if $m$ is constant, then there is a constant probability that no proportional allocation exists.

\begin{theorem} \label{thm:prop-lb}
For any $(\alpha, \beta)$-PDF-bounded distribution $\cD$ and any $m, n$ such that $\frac{m}{n} \leq \frac{1}{2\beta}$, the probability that no proportional allocation exists is at least $e^{-2\beta m}$. 
\end{theorem}

\Cref{thm:prop-ub-main,thm:prop-lb} essentially complete our understanding of the asymptotic existence of proportional allocations for chores. 
Again, this demonstrates a stark difference from the case of goods: for goods, no proportional allocation exists if $m < n$, whereas for chores, we have existence even when $m$ is much smaller than $n$ (e.g., $m = \Theta(\log\log n)$).

Overall, our findings indicate that achieving fairness is, interestingly, \emph{easier} for chores than for goods when viewed from an asymptotic perspective.

\subsection{Overview of Proofs}
\label{sec:overview}

We now provide a high-level overview of our proofs, starting with the existence of envy-free allocations \textbf{(\Cref{thm:upper-main})}. 
The proof of this theorem involves handling three cases.
\begin{itemize}
\item \underline{Case I}: 
$m = \Omega(n \log n)$. 
In this case, we can use a similar approach to \citet{DickersonGoKa14}, namely, we assign each chore to an agent with the lowest disutility for it, thereby minimizing social cost. 
A standard application of concentration inequalities shows that this yields an envy-free allocation with high probability.
\item \underline{Case II}: 
$m = O(n \log n)$ and $m$ is divisible by $n$. 
In this case, the algorithm of \citet{ManurangsiSu20} for goods directly yields an envy-free allocation for chores. 
The reason is that their algorithm always outputs a \emph{balanced} allocation, that is, an allocation in which every agent receives the same number of items. 
We may therefore view each chore as a ``good'' whose utility is a complement of the disutility of the chore; this transformation preserves the envy-freeness of allocations.
\item \underline{Case III}: 
$m = O(n \log n)$ and $m$ is not divisible by $n$. 
This is the key novelty we contribute to the proof. 
Our algorithm starts by running the algorithm from Case~II on the first $\lfloor m/n \rfloor \cdot n$ chores. 
We then show that most of the agents have a sufficiently large ``gap'' between their own disutility and their disutility for any other agent's bundle.
For these agents, such a gap allows us to assign the remaining chores to them as long as the disutilities are sufficiently low; we use a matching-based algorithm for this second step. 
It is worth noting that this case precisely captures the difference between goods and chores: 
for goods, such a gap would \emph{not} allow us to assign an additional good to the agent, as some other agent might envy this agent as a result.
\end{itemize}

As for the non-existence of envy-free allocations \textbf{(\Cref{thm:lb-main})}, we take inspiration from the corresponding result of \citet{DickersonGoKa14}.
In particular, we consider chores that are ``favorite'' (i.e., yield the lowest disutility) for multiple agents. 
We observe that each of these chores cannot be in a singleton bundle in an envy-free allocation; otherwise, one such agent would envy the agent who receives the chore. 
By establishing a lower bound on the number of such chores, we arrive at our result.

We remark that there are two important differences between our non-existence result and that of \citet{DickersonGoKa14}. 
Firstly, in the case of goods, if a good is the favorite good for $k$ agents, then at least $k - 1$ of the agents must receive at least two goods in order to be envy-free. 
This is not true for chores, so we use the aforementioned observation instead.
Secondly, to derive a high probability bound, Dickerson \textit{et al.}~relied on an application of Markov's inequality---this leads to a lower bound of only $m \geq n + o(n)$. 
In this paper, we utilize a more refined approach based on the so-called Efron--Stein inequality, which allows us to obtain a high probability bound even for $m \geq n + \Theta(n)$. 
We note that our approach can also be applied to the case of goods; however, we do not provide the details as it does not give us the tight bound (and the tight bound in that case has already been established by \citet{ManurangsiSu20}).

Next, we consider proportionality, again starting with existence \textbf{(\Cref{thm:prop-ub-main})}.
Observe that if $m \geq 2n$, the existence of envy-free allocations (\Cref{thm:upper-main}) immediately implies that of proportional allocations. 
Thus, we may focus on the setting where $m < 2n$. 
It is further useful to distinguish between two cases: the ``small-$m$'' case where $m = O(n/\log n)$ and the ``medium-$m$'' case where $m = \Omega(n / \log n)$. 
In the small-$m$ case, we restrict our attention to allocations in which each agent receives either one or no chore. 
With this restriction, we can create a bipartite graph that indicates whether a chore can be assigned to an agent---that is, a graph with the agents on one side and the chores on the other side, such that there is an edge between an agent and a chore exactly when the chore can be assigned to the agent without violating proportionality---and then find a right-saturated matching in this graph. 
However, this algorithm is somewhat challenging to analyze, as the edges are dependent and we cannot simply apply results from random graph theory (e.g.,~\citep{ErdosRe64}). 

To mitigate this technical challenge, we modify the algorithm slightly by considering the subgraph of the bipartite graph where each agent is (potentially) connected to only her favorite chore. 
The advantage here is that, since every vertex on the left has degree at most one, a right-saturated matching exists as long as no vertex on the right is isolated. 
The latter event is simple and the probability that it occurs can be conveniently bounded.

For the medium-$m$ case, the strategy above does not quite work, as there is a constant probability that some vertex on the right will be isolated. 
To circumvent this issue, we partition the chores into subsets so that we reduce the problem to the small-$m$ case. 
We then run the previously described algorithm on each subset and let the final allocation be the union of the allocations output on these subsets.

Finally, we address the non-existence of proportional allocations \textbf{(\Cref{thm:prop-lb})}.
Our proof is based on the following observation: 
if a chore has disutility larger than $m/n$ for an agent $i$, then it cannot be assigned to $i$ in a proportional allocation. 
We then show that for any chore, such an event occurs for \emph{all} agents $i$ with a constant probability---when this happens, no proportional allocation exists.

\section{Preliminaries}
\label{sec:prelim}

Let $N = \{1, \dots, n\}$ be the set of agents and $M = \{1, \dots, m\}$ be the set of chores, where $n,m\ge 2$. 
For each $(i,j)\in N\times M$, we use $d_i(j)\ge 0$ to denote the \emph{disutility} (or \emph{cost}) of chore~$j$ for agent~$i$. 
We assume throughout the paper that the disutilities are additive, that is, $d_i(M') = \sum_{j \in M'} d_i(j)$ for all $M' \subseteq M$.

An allocation $A = (A_1, \dots, A_n)$ is a partition of $M$ into $n$ disjoint subsets, where the bundle $A_i$ is assigned to agent~$i$. 
We say that an allocation $A$ is \emph{envy-free} if for every pair of agents $i, i'$, it holds that $d_i(A_i) \leq d_i(A_{i'})$. 
When this condition does not hold for some pair $i,i'$ (i.e., $d_i(A_i) > d_i(A_{i'})$), we say that $i$ envies $i'$ (with respect to the allocation $A$).
We say that an allocation $A = (A_1, \dots, A_n)$ is \emph{proportional} if $d_i(A_i) \leq d_i(M)/n$ for all $i \in N$.

We assume that each disutility $d_i(j)$ is drawn independently at random from a distribution $\cD$. 
The support of $\cD$ lies in $[0, 1]$ and $\cD$ is \emph{non-atomic}, i.e., $\Pr_{X \sim \cD}[X = x] = 0$ for every $x \in [0, 1]$. 
For the sake of brevity, we will not repeat these assumptions in our formal statements.
Due to non-atomicity, with probability~$1$, all disutilities are positive (i.e., $d_i(j) > 0$ for all $i\in N$ and $j\in M$) and distinct (i.e., $d_i(j) \ne d_{i'}(j')$ for all $i,i'\in N$ and $j,j'\in M$ such that $(i,j) \ne (i',j')$); in particular, each agent has a unique ``favorite chore'' (i.e., the set $\argmin_{j \in M} d_i(j)$ has size~$1$).
We will assume that this is the case throughout the paper, and slightly abuse notation by using $\argmin_{j \in M} d_i(j)$ to denote an item rather than a set.
In addition, for most of our results, we assume that $\cD$ is \emph{PDF-bounded}. 
This means that there exist constants $\alpha, \beta > 0$ such that $f_{\cD}(x) \in [\alpha, \beta]$ for all $x \in [0, 1]$, where $f_{\cD}$ denotes the probability distribution function of $\cD$.
We sometimes write \emph{$(\alpha,\beta)$-PDF-bounded} if we want to specify the constants $\alpha,\beta$ explicitly.
Notice that $\alpha \le 1$ and $\beta \ge 1$ hold for every distribution $\cD$; the uniform distribution on $[0,1]$ has $\alpha = \beta = 1$ (and is the only distribution with this property).

We use $\log$ to denote the natural logarithm, and say that an event occurs ``with high probability'' if the probability that it occurs approaches $1$ as $n\rightarrow\infty$.

\subsection{Concentration Inequalities}

We now list two concentration inequalities that we will make use of.
The first is the well-known Chernoff bound.

\begin{lemma}[Chernoff bound] \label{lem:chernoff}
Let $X_1, \dots, X_k$ be independent random variables that are bounded in $[0, 1]$, and let $X := X_1 + \cdots + X_k$. 
Then, for any $\delta \ge 0$,
$$\Pr[X \geq (1 + \delta)\E[X]] \leq \exp\left(\frac{-\delta^2 \E[X]}{2 + \delta}\right)$$
and
$$\Pr[X \leq (1 - \delta)\E[X]] \leq \exp\left(\frac{-\delta^2 \E[X]}{2}\right).$$
\end{lemma}

The second inequality is the so-called Efron--Stein inequality (e.g., \cite[Corollary 3.2]{BoucheronLuMa13}).
Given a set $\cZ$ and a real number $c\ge 0$, we say that a function $F: \cZ^n \to \R$ is \emph{$c$-difference bounded} if for all $i\in\{1,\dots,n\}$ and $z_1, \dots, z_n, z'_i \in \cZ$, it holds that
\begin{align*}
|F(z_1, &\dots, z_{i-1}, z_i, z_{i+1}, \dots, z_n) \\
&- F(z_1, \dots, z_{i-1}, z'_i, z_{i+1}, \dots, z_n)| \leq c.
\end{align*}

\begin{lemma}[Efron--Stein inequality] \label{lem:es-ineq}
Let $c\ge 0$ be a real number, and let $Z_1, \dots, Z_n$ be independent random variables from sample space $\cZ$.
If $F: \cZ^n \to \R$ is a $c$-difference bounded function, then
\begin{align*}
\var(F(Z_1, \dots, Z_n)) \leq \frac{nc^2}{4}. 
\end{align*}
\end{lemma}

\subsection{Matchings in Random Graphs}

Next, we state some definitions and results on matchings in random graphs.

\begin{definition}
Let $G$ be a bipartite graph and $r$ be a positive integer.
An \emph{$r$-matching} of $G$ is a subgraph obtained by removing some edges from $G$ so that every left vertex has degree at most $r$ and every right vertex has degree at most $1$. 
When $r = 1$, we simply call an $r$-matching a \emph{matching}.

An ($r$-)matching is said to be \emph{right-saturated} if every right vertex has degree exactly $1$. 
It is said to be \emph{perfect} if every left vertex has degree exactly $r$ and every right vertex has degree exactly $1$. 
\end{definition}

Recall that an \emph{Erd{\H{o}}s-R{\'{e}}nyi random bipartite graph} is a random graph in which each edge between a pair of left and right vertices is included with probability $p$, independently of other edges.
We write $\cG(n_L, n_R, p)$ to denote the distribution of such a random graph with $n_L$ left vertices and $n_R$ right vertices.
A classic result, due to Erd{\H{o}}s and R{\'{e}}nyi themselves, is that a perfect matching in $\cG(n, n, p)$ exists with high probability as long as $p$ is noticeably above $\log n / n$.

\begin{lemma}[\citet{ErdosRe64}] \label{lem:er-matching}
Suppose that $G \sim \cG(n, n, p)$ where $p = (\log n + \omega(1))/n$. 
Then, with high probability, $G$ contains a perfect matching. 
\end{lemma}

In this work, we will use the high-probability existence of a right-saturated $2$-matching in the setting where the right side has at most twice as many vertices as the left side. 

\begin{lemma} \label{lem:matching}
Suppose that $G \sim \cG(n_L, n_R, p)$ where $n/2 \leq n_L \leq n$, $n_R \leq n$, and $p = 2(\log n + \omega(1)) / n$. 
Then, with high probability, $G$ contains a right-saturated $2$-matching. 
\end{lemma}

\begin{proof}
Let $G = (L, R, E)$. We consider two cases.
\begin{itemize}
\item \underline{Case I}: $n_R \leq n_L$. 
In this case, let $R'$ be a set of $n_L - n_R$ additional right vertices. 
Consider a graph $G' = (L, R \cup R', E \cup E')$, where we include each pair $(\ell, r') \in L \times R'$ in $E'$  independently with probability $p$. 
We have $G' \sim \cG(n_L, n_L, p)$. 
Thus, by \Cref{lem:er-matching}, a perfect matching in $G'$ exists with high probability.\footnote{Recall that $n_L \in [n/2, n]$; this means that $n \to \infty$ is equivalent to $n_L \to \infty$. 
Since $p = 2(\log n + \omega(1)) / n \geq (\log n_L + \omega(1)) / n_L$, \Cref{lem:er-matching} ensures that the matching exists with high probability.} 
This immediately implies that a right-saturated ($2$-)matching in $G$ exists with high probability.
\item \underline{Case II}: 
$n_R > n_L$. 
From the conditions on $n_L$ and $n_R$, we have $n_R \leq 2 \cdot n_L$.
Partition $R$ into two sets $R_1, R_2$ where $R_1$ has $n_L$ vertices. 
The graph $G_1 = (L, R_1, E)$ is distributed according to $\cG(n_L, n_L, p)$; thus, \Cref{lem:er-matching} implies that a perfect matching exists in $G_1$ with high probability.
Furthermore, by a similar reasoning to Case I, we have that a right-saturated matching exists in $G_2 = (L, R_2, E)$ with high probability. 
By combining these two matchings, we conclude that a right-saturated $2$-matching exists in $G$ with high probability.
\end{itemize}
The two cases together complete the proof.
\end{proof}

\section{Envy-Freeness}

In this section, we prove our results on envy-freeness (\Cref{thm:upper-main,thm:lb-main}).

\subsection{Non-Existence}

We start with the non-existence (\Cref{thm:lb-main}).
For each agent $i\in N$, let $Z_i$ denote $i$'s favorite chore.
For each chore $j\in M$, let $N^{\fav}_j$ denote the set of agents whose favorite chore is $j$, i.e., $N^{\fav}_j := \{i \in N \mid Z_i = j\}$. 
We say that a chore $j$ is a \emph{repeated favorite chore} if it is a favorite chore of more than one agent, i.e., $|N^{\fav}_j| > 1$.
Let $T$ denote the number of repeated favorite chores.
The following lemma relates $T$ to whether an envy-free allocation can exist.

\begin{lemma} \label{lem:lb-from-repeated-chores}
If the number of repeated favorite chores is more than $2(m - n)$, then no envy-free allocation exists.
\end{lemma}

\begin{proof}
Suppose for contradiction that there exists an envy-free allocation with more than $2(m-n)$ repeated favorite chores.
Since every disutility is positive, no agent can receive an empty bundle.
If at most  $2n - m - 1$ chores are in a bundle of size one, then the total number of chores is at least 
\begin{align*}
&(2n-m-1) + 2(n-(2n-m-1)) \\
&= (2n-m-1) + 2(m+1-n) = m+1,
\end{align*}
a contradiction.
Therefore, at least $2n - m$ chores must be in a bundle of size one. 
If no repeated favorite chore is in a bundle of size one, then the total number of chores is more than $(2n-m) + 2(m-n) = m$, a contradiction.
Hence, there is a repeated favorite chore in a bundle of size one.
This means that there exists an agent who does not receive this chore but for whom the chore is her favorite chore.
This agent will therefore envy the agent who receives this chore, yielding the desired contradiction.
\end{proof}

Next, we calculate the expectation of $T$, the number of repeated favorite chores.

\begin{lemma} \label{lem:num-repeated-expectation}
$\E[T] \geq m \left(1 - \left(1 + \frac{n-1}{m}\right) e^{-\frac{n-1}{m}}\right)$.
\end{lemma}

\begin{proof}
By linearity of expectation, we have 
\begin{align*}
\E[T] 
&= \E\left[\sum_{j \in M} \ind[|N^{\fav}_j| > 1] \right] \\
&= \sum_{j \in M}\E\left[ \ind[|N^{\fav}_j| > 1] \right] \\
&= \sum_{j \in M} \Pr[|N^{\fav}_j| > 1] \\
&= m - \sum_{j \in M} \left(\Pr[|N^{\fav}_j| = 0] + \Pr[|N^{\fav}_j| = 1]\right) \\
&= m - \sum_{j \in M} \Bigg(\Pr[Z_1, \dots, Z_n \ne j] + \sum_{i \in N} \Pr[Z_i = j \\
&\qquad\qquad\qquad\wedge Z_1, \dots, Z_{i-1}, Z_{i+1}, \dots, Z_n \ne j]\Bigg).
\end{align*}
Observe that $Z_1, \dots, Z_n$ are distributed uniformly at random from $M$. Thus, we have
\begin{align*}
\E[T] &= m - \sum_{j \in M} \left(\left(1 - \frac{1}{m}\right)^n + n \cdot \frac{1}{m} \cdot \left(1 - \frac{1}{m}\right)^{n - 1} \right) \\
&= m \left(1 - \left(1 + \frac{n - 1}{m}\right)\left(1 - \frac{1}{m}\right)^{n - 1}\right) \\
&\geq m \left(1 - \left(1 + \frac{n - 1}{m}\right) e^{-\frac{n-1}{m}}\right),
\end{align*}
where the inequality follows from the estimate $1 + x \le e^x$, which holds for all real numbers $x$.
\end{proof}

With all the ingredients in place, we are ready to prove \Cref{thm:lb-main}.

\begin{proof}[Proof of \Cref{thm:lb-main}]
Let $\tnu = m/(n-1)$ and 
\[
f(x) = 2 - x\left(1 + \left(1 + \frac{1}{x}\right)e^{-\frac{1}{x}}\right).
\]
We have 
\[
f'(x) = -\frac{e^{-1/x}((e^{1/x}+1)x^2+x+1)}{x^2},
\]
so $f$ is decreasing on $(0, \infty)$.
Since $m/n \leq \nu - \eps$, it holds that $\tnu \leq \nu - \eps/2$ when $n$ is sufficiently large. 
By \Cref{lem:num-repeated-expectation}, we have
\begin{align} \label{eq:exp-diff-lb}
&\E[T] - 2(m - n) \nonumber \\
&\geq m \left(1 - \left(1 + \frac{n-1}{m}\right) e^{-\frac{n-1}{m}}\right) - 2(m - n) \nonumber \\
&= (n-1) \cdot f(\tnu) + 2 
\geq (n-1) \cdot f\left(\nu - \eps/2\right).
\end{align}
Since $f$ is decreasing on $(0, \infty)$ and $f(\nu) = 0$, we deduce that $f(\nu - \eps/2)$ is a positive constant, and $\E[T] \ge 2(m-n)$.

Observe that $T$ is a function of $Z_1, \dots, Z_n$ which is $1$-bounded (i.e., changing a single $Z_i$ changes the value of $T$ by at most one). 
By \Cref{lem:es-ineq}, we have $\var(T) \leq n/4$. 
Applying Chebyshev's inequality, we get
\begin{align*}
\Pr[T \leq 2(m - n)] 
&\leq \frac{\var(T)}{\left(\E[T] - 2(m - n)\right)^2} \\
&\overset{\eqref{eq:exp-diff-lb}}{\leq} \frac{n}{4 \cdot (n-1)^2 \cdot f(\nu - \eps/2)^2} = O\left(\frac{1}{n}\right). 
\end{align*}
Finally, by \Cref{lem:lb-from-repeated-chores}, this implies that with high probability, no envy-free allocation exists.
\end{proof}

\subsection{Existence}

We now proceed to the existence (\Cref{thm:upper-main}). 
As discussed in \Cref{sec:overview}, this is based on two theorems depending on whether $m$ is ``large'' relative to $n$.

\begin{theorem} \label{thm:welfare-max}
For any PDF-bounded distribution $\cD$, there exists a constant $c$ such that, if $m \geq c n \log n$, then there exists a polynomial-time algorithm that finds an envy-free allocation with high probability.
\end{theorem}

\begin{theorem} \label{thm:indiv-matching}
For any PDF-bounded distribution $\cD$, if $2n \leq m \leq n^{8/7}$, then there exists a polynomial-time algorithm that finds an envy-free allocation with high probability.
\end{theorem}

Note that \Cref{thm:upper-main} is an immediate consequence of these two theorems.

For \Cref{thm:welfare-max}, we use the ``cost-minimizing algorithm'', which allocates each chore $j$ to the agent with the smallest disutility for it, that is, the agent $\argmin_{i \in N} d_i(j)$.
Clearly, this algorithm runs in polynomial time.
The proof of its correctness is an adaptation of the analogous proof for the welfare-maximizing algorithm in the case of goods \citep{DickersonGoKa14}.

For every distinct $i, i' \in N$ and $j \in M$, let us define the following random variables:
\begin{align*}
X_i^j &=
\begin{cases}
d_i(j) &\text{ if } i = \argmin_{i_0 \in N} d_{i_0}(j); \\
0 &\text{ otherwise;}
\end{cases} \\
Y_{i, i'}^j &=
\begin{cases}
d_i(j) &\text{ if } i' = \argmin_{i_0 \in N} d_{i_0}(j); \\
0 &\text{ otherwise.}
\end{cases}
\end{align*}
Let $A = (A_1, \dots, A_n)$ be the output allocation. 
Observe that $d_i(A_i) = \sum_{j \in M} X_i^j$ and $d_i(A_{i'}) = \sum_{j \in M} Y_{i, i'}^j$.

Let $\mu$ be the mean of $\cD$, $\mu^* := \E[X_i^j]$, and $\mu' := \E[Y_{i, i'}^j]$ for some fixed $i,i',j$. 
We prove the following lemma.\footnote{Note that our proof implicitly shows that $\mu \geq \var(\cD)$.}

\begin{lemma} \label{lem:mean-bound-welfare-max}
$\mu^* \leq (\mu - \var(\cD)) / n$ and $\mu' \geq \mu / n$.
\end{lemma}

\begin{proof}
To bound $\mu^*$, note that since $d_1(j), \dots, d_n(j)$ are independent and identically distributed,
\begin{align*}
\mu^* 
= \E[X_i^j] 
&= \Pr[i = \argmin_{i_0 \in N} d_{i_0}(j)] \\
&\quad \cdot \E[X_i^j \mid i = \argmin_{i_0 \in N} d_{i_0}(j)] \\
&= \frac{1}{n} \cdot \E[\min\{d_1(j), \dots, d_n(j)\}]. 
\end{align*}
We can bound the last term further as follows:
\begin{align*}
&\E[\min\{d_1(j), \dots, d_n(j)\}] \\
&\leq \E[\min\{d_1(j), d_2(j)\}] \\
&= \E\left[\frac{d_1(j) + d_2(j)}{2} - \frac{|d_1(j) - d_2(j)|}{2}\right] \\
&= \mu - \frac{1}{2} \cdot \E\left[|d_1(j) - d_2(j)|\right] \\
&\leq \mu - \frac{1}{2} \cdot \E\left[(d_1(j) - d_2(j))^2\right] \\
&= \mu - \var(\cD),
\end{align*}
where the second inequality holds because the disutilities are in $[0, 1]$.
This gives the desired bound on $\mu^*$.

Now, by symmetry, we have
\begin{align*}
\mu = \E[d_i(j)] &= \E\left[X_i^j + \sum_{i'' \in N \setminus \{i\}} Y_{i, i''}^j\right] \\
&= \mu^* + (n - 1) \mu'.
\end{align*}
It follows that $\mu' = \frac{\mu - \mu^*}{n - 1} \geq \frac{\mu}{n}$, as desired.
\end{proof}

We are ready to prove \Cref{thm:welfare-max}.

\begin{proof}[Proof of \Cref{thm:welfare-max}]
Since $\cD$ is PDF-bounded and non-atomic, its variance is finite and positive. 
Let $\delta_0 = \frac{0.5 \var(\cD)}{\mu}$ and $c = \frac{100}{\delta_0^2\left(\mu - 0.5 \cdot \var(\cD)\right)}$, and suppose that $m \geq c n \log n$.
Since $\var(\cD) \le \mu$, we have $\delta_0 \le 0.5$.

Fix any distinct $i, i' \in N$. 
Recall that $d_i(A_i) = \sum_{j \in M} X_i^j$ is a sum of $m$ independent random variables whose values are bounded in $[0, 1]$. 
Let $\delta = \frac{\mu - 0.5 \cdot \var(\cD)}{n\mu^*} - 1$.
We have
\begin{align*}
\mu(\mu-0.5\var(\cD))
&= \mu^2 - 0.5\var(\cD)\mu \\
&\ge \mu^2 - 0.5\var(\cD)\mu - 0.5\var(\cD)^2 \\
&= (\mu + 0.5\var(\cD))(\mu - \var(\cD)) \\
&\ge (\mu + 0.5\var(\cD))\cdot n\mu^*,
\end{align*}
where the last inequality follows from \Cref{lem:mean-bound-welfare-max}.
Dividing both sides by $n\mu\mu^*$ and rearranging, we obtain
\begin{align*}
\frac{\mu - 0.5\var(\cD)}{n\mu^*} - 1 \ge \frac{0.5\var(\cD)}{\mu},
\end{align*}
which implies that $\delta \ge \delta_0$.

Applying \Cref{lem:chernoff} with $\delta$, we get
\begin{align*}
&\Pr\left[d_i(A_i) \geq \frac{m}{n}(\mu - 0.5 \cdot \var(\cD))\right] \\
&\leq \exp\left(-\frac{\delta^2}{(1 + \delta)(2 + \delta)} \cdot \frac{m}{n}\left(\mu - 0.5 \cdot \var(\cD)\right)\right) \\
&\leq  \exp\left(-\frac{(\delta_0)^2}{(1 + \delta_0)(2 + \delta_0)} \cdot \frac{m}{n}\left(\mu - 0.5 \cdot \var(\cD)\right)\right) \\
&\leq \exp(-10 \log n) \\
&= \frac{1}{n^{10}}.
\end{align*}
Here, the second inequality follows from the fact that the function $\frac{x^2}{(1 + x)(2 + x)}$ is increasing on $(0, \infty)$ (as its derivative is $\frac{x(3x+4)}{(1+x)^2(2+x)^2} > 0$) and $\delta \geq \delta_0$, and the third inequality from our choice of $c$ and the fact that $\delta_0 \le 0.5$.

Similarly, let $\delta' = 1 - \frac{\mu - 0.5 \cdot \var(\cD)}{n\mu'}$.
By \Cref{lem:mean-bound-welfare-max}, we have
\begin{align*}
n\mu'(\mu-0.5\var(\cD))
&\ge \mu(\mu-0.5\var(\cD)).
\end{align*}
Dividing both sides by $n\mu\mu'$ and rearranging, we get
\begin{align*}
1 - \frac{\mu - 0.5\var(\cD)}{n\mu'} \ge \frac{0.5\var(\cD)}{\mu},
\end{align*}
which implies that $\delta' \ge \delta_0$.

Since $d_i(A_{i'}) = \sum_{j \in M} Y_{i, i'}^j$, applying \Cref{lem:chernoff} with $\delta'$, we get 
\begin{align*}
&\Pr\left[d_i(A_{i'}) \leq \frac{m}{n}(\mu - 0.5 \cdot \var(\cD))\right] \\
&\leq \exp\left(-\frac{(\delta')^2}{2} \cdot m\mu' \right) \\
&\leq \exp\left(-\frac{(\delta')^2}{2} \cdot \frac{m}{n} \cdot \mu\right) \\
&\leq \exp\left(-\frac{(\delta_0)^2}{2} \cdot \frac{m}{n} \cdot \mu\right) \leq \exp(-10 \log n) = \frac{1}{n^{10}},
\end{align*}
where the second inequality follows from \Cref{lem:mean-bound-welfare-max}, the third inequality from the fact that $\delta' \ge \delta_0$, and the last inequality from our choice of $c$.

Thus, by the union bound, we have $d_i(A_i) \ge d_i(A_{i'})$ with probability at most $2/n^{10}$.
Taking the union bound over all distinct $i, i' \in N$, we find that the probability that some agent envies some other agent with respect to the allocation~$A$ is at most $2/n^8 = o(1)$.
It follows that $A$ is envy-free with high probability, as desired.
\end{proof}

Next, we will establish \Cref{thm:indiv-matching} via a matching-based approach.
We start with the case where $m$ is divisible by $n$.
We say that an allocation $A = (A_1, \dots, A_n)$ is \emph{balanced} if $|A_1| = \dots = |A_n|$. 
When $m = rn$ for some integer $r \geq 2$, the algorithm by \citet{ManurangsiSu20} finds a balanced allocation that is envy-free for the case of goods. 
It is not difficult to see that this implies a similar algorithm for chores, as we may simply set the ``utility'' $u_i(j)$ to be $1 - d_i(j)$; a balanced allocation is envy-free with respect to these utilities if and only if it is envy-free with respect to the disutilities. 
Moreover, their algorithm ensures that every agent's total utility for the goods that she receives is at least $r\left(1  - O\left(\frac{\log(rn)}{n}\right)\right)$; in the chore setting, this translates to a total disutility of at most $O\left(\frac{r \log(rn)}{n}\right)$. 
We summarize these properties more formally below.

\begin{theorem}[\citet{ManurangsiSu20}] \label{thm:divisible}
Suppose that $m = rn$ for some integer $r$ such that $2 \leq r \leq e^{n^{0.1}}$. 
For any $(\alpha, \beta)$-PDF-bounded distribution $\cD$, there exists a polynomial-time algorithm $\algdiv$ that, with high probability, finds a balanced envy-free allocation $A = (A_1, \dots, A_n)$ such that $d_i(A_i) \leq \frac{c r \log(rn)}{n}$ for all $i\in N$, where $c > 0$ is a constant depending only on $\alpha, \beta$.
\end{theorem}

We now proceed to the case where $m$ is not divisible by $n$. 
Recall that $\algdiv$ is the algorithm from the divisible case in \Cref{thm:divisible}. 
Our algorithm for the indivisible case is presented as \Cref{alg:two-stage-matching}, where $\tau$ is a parameter that we will choose below. 
Note that the algorithm runs in polynomial time since $\algdiv$ runs in polynomial time and a right-saturated $2$-matching can be found by creating two copies of each left vertex and finding a maximum matching.

\begin{algorithm}
\caption{Matching for Envy-Free Allocation}
\label{alg:two-stage-matching}
\begin{algorithmic}[1]
\Procedure{TwoStageAlgorithm$_{\tau}(N, M, \{d_i\}_{i\in N})$}{}
\State $r = \lfloor m/n \rfloor$
\State $M^0 \gets $ the first $rn$ chores
\State $A^0 = (A^0_1, \dots, A^0_n) \gets \algdiv(M^0)$
\State $M^1 \gets M \setminus M^0$
\State $N^1 \gets \{i \in N \mid d_i(A^0_i) \leq \min_{i' \in N \setminus \{i\}} d_i(A^0_{i'})  - 2\tau\}$
\State $E = \{(i, j) \in N^1 \times M^1 \mid d_i(j) \leq \tau\}$
\If{a right-saturated $2$-matching exists in the graph $G = (N^1, M^1, E)$}
\State $A^1 = (A^1_1, \dots, A^1_n) \gets$ the allocation corresponding to a right-saturated $2$-matching
\State \Return $A = (A^0_1 \cup A^1_1, \dots, A^0_n \cup A^1_n)$
\Else 
\State \Return NULL
\EndIf
\EndProcedure
\end{algorithmic}
\end{algorithm}

\begin{proof}[Proof of \Cref{thm:indiv-matching}]
Let $\cD$ be $(\alpha,\beta)$-PDF-bounded, and let $r := \lfloor m/n\rfloor \le n^{1/7}$.
We use \Cref{alg:two-stage-matching} with $\tau = \frac{3\log n}{\beta n}$.
For convenience, let $\xi := \frac{cr \log (rn)}{n}$, where $c$ is the constant in \Cref{thm:divisible}.
From that theorem, with high probability, the output (partial) allocation $A^0 = (A^0_1, \dots, A^0_n)$ satisfies the following conditions with high probability:
\begin{align}
d_i(A^0_i) \leq d_i(A^0_{i'}) & &\forall i, i' \in N; \label{eq:ef-partial} \\
d_i(A^0_i) \leq \xi & &\forall i \in N. \label{eq:val-ef}
\end{align}
For the remainder of this proof, we will assume that these conditions are satisfied.

We claim that if \Cref{alg:two-stage-matching} does not return NULL, then the output is envy-free. 
To see this, fix any distinct $i, i' \in N$ and consider two cases.
\begin{itemize}
\item \underline{Case I}: $i \notin N^1$. In this case, we have $$d_i(A_i) = d_i(A^0_i) \overset{\eqref{eq:ef-partial}}{\leq} d_i(A^0_{i'}) \leq d_i(A_{i'}).$$
\item \underline{Case II}: $i \in N^1$. In this case, we have
\begin{align*}
d_i(A_i) = d_i(A^0_i) + d_i(A^1_i) &\leq d_i(A^0_i) + 2\tau \\
&\leq d_i(A^0_{i'}) \leq d_i(A_{i'}),
\end{align*}
where the first inequality follows from the definition of~$E$ and the fact that $|A^1_i| \leq 2$, and the second inequality from the definition of $N^1$.
\end{itemize}
Thus, it suffices to show that a right-saturated $2$-matching exists in $G$ with high probability. 
To this end, we will prove that
\begin{align} \label{eq:LHS-size-whp}
\Pr[|N^1| < n/2] = o(1).
\end{align}
Before we prove \eqref{eq:LHS-size-whp}, let us argue why it gives us the desired result. 
Notice that for fixed $N^1$, the graph $G$ is drawn according to $\cG(|N^1|, m - rn, p)$ with $p = \Pr_{X \sim \cD}[X \leq \tau] \leq \beta \tau = 3 \log n/n$, where the inequality follows from the $(\alpha, \beta)$-PDF-boundedness of $\cD$ and the last equality from our choice of $\tau$. 
As a result, if $|N^1| \geq n/2$, \Cref{lem:matching} implies that a right-saturated $2$-matching exists in $G$ with high probability. 
From this and \eqref{eq:LHS-size-whp}, letting $Q$ denote the event that $G$ has a right-saturated $2$-matching, we have 
\begin{align*}
\Pr[Q] 
&\geq \Pr[Q \text{ and } |N^1| \geq n/2] \\
&= \Pr[Q \mid |N^1| \geq n/2]\cdot \Pr[|N^1| \geq n/2] \\
&= \Pr[Q \mid |N^1| \geq n/2] \\
&\quad- \Pr[Q \mid |N^1| \geq n/2]\cdot \Pr[|N^1| < n/2] \\
&\geq \Pr[Q \mid |N^1| \geq n/2] - \Pr[|N^1| < n/2]. \\
&\geq 1 - o(1),
\end{align*}
as desired.

We are now left to prove \eqref{eq:LHS-size-whp}. First, using Markov's inequality, \eqref{eq:val-ef}, and the union bound, we have
\begin{align}
&\Pr[|N^1| < n/2] \nonumber \\
&= \Pr[|N \setminus N^1| > n/2] \nonumber \\
&\leq \frac{2}{n} \cdot \E[|N \setminus N^1|] \nonumber \\
&= \frac{2}{n} \cdot \left(\sum_{i \in N} \Pr[i \notin N^1]\right) \nonumber \\
&\leq \frac{2}{n} \cdot \left(\sum_{i \in N} \Pr[\exists i' \in N \setminus \{i\}, d_i(A^0_{i'}) < \xi + 2\tau]\right) \nonumber \\
&\leq \frac{2}{n} \cdot \left(\sum_{i \in N} \sum_{i' \in N \setminus \{i\}} \Pr[d_i(A^0_{i'}) < \xi + 2\tau]\right). \label{eq:expand-small-LHS-prob}
\end{align}

Next, for any $\gamma > 0$ and $i\in N$, let $M^0_{i, \leq \gamma}$ denote the set $\{j \in M^0 \mid d_i(j) \leq \gamma\}$.
From \eqref{eq:val-ef}, for any $i'\in N$, we have $A^0_{i'} \subseteq M^0_{i', \leq \xi}$. 
Meanwhile, if $d_i(A^0_{i'}) < \xi + 2\tau$, then it holds that $A^0_{i'} \subseteq M^0_{i, \leq \xi + 2\tau}$; together with the previous sentence, this implies that $|M^0_{i', \leq \xi} \cap M^0_{i, \leq \xi + 2\tau}| \geq r$. 
Hence, for any distinct $i, i' \in N$, we have
\begin{align*}
&\Pr[d_i(A^0_{i'}) < \xi + 2\tau] \\
&\leq \Pr[|M^0_{i', \leq \xi} \cap M^0_{i, \leq \xi + 2\tau}| \geq r] \\
&\leq \Pr[|M^0_{i', \leq \xi} \cap M^0_{i, \leq \xi + 2\tau}| \geq 2] \\
&= \Pr[\exists \text{ distinct }j, j' \in M^0, j, j' \in M^0_{i', \leq \xi} \cap M^0_{i, \leq \xi + 2\tau}] \\
&\leq \sum_{\text{distinct }j, j' \in M^0} \Pr[j, j' \in M^0_{i', \leq \xi} \cap M^0_{i, \leq \xi + 2\tau}],
\end{align*}
where the last inequality follows from the union bound.

Notice that each chore belongs to $M^0_{i', \leq \xi}$ and $M^0_{i, \leq \xi + 2\tau}$ independently with probability $\Pr_{X \sim \cD}[X \leq \xi] \leq \beta \xi$ and $\Pr_{X \sim \cD}[X \leq \xi + 2 \tau] \leq \beta(\xi + 2 \tau)$, respectively. 
Plugging this into the estimate above, we get
\begin{align*}
\Pr[d_i(A^0_{i'}) < \xi + 2\tau]  
&\leq (rn)^2 \cdot \left(\beta \xi\right)^2\left(\beta(\xi + 2\tau)\right)^2 \\
&\leq O\left(\frac{r^6 \log^4(rn)}{n^2}\right),
\end{align*}
where the latter inequality follows from $\xi, \tau \leq O\left(\frac{r \log(rn)}{n}\right)$. 
Plugging this back into \eqref{eq:expand-small-LHS-prob}, we obtain
\begin{align*}
\Pr[|N^1| < n/2] 
&\leq O\left(\frac{1}{n} \cdot n^2 \cdot \frac{r^6 \log^4(rn)}{n^2}\right) \\
&= O\left(\frac{r^6 \log^4(rn)}{n}\right),
\end{align*}
which is $o(1)$ due to our assumption that $r \leq n^{1/7}$.
\end{proof}

\section{Proportionality}

In this section, we turn our attention to proportionality (\Cref{thm:prop-ub-main,thm:prop-lb}).

\subsection{Non-Existence}

We start with our non-existence result, whose simple proof follows precisely the outline in \Cref{sec:overview}.

\begin{proof}[Proof of \Cref{thm:prop-lb}]
Consider the first chore~$1$. 
Notice that if $d_i(1) > m/n$ for all agents $i \in N$, then no proportional allocation exists, since the agent who is assigned this chore will incur disutility greater than $m/n \geq d_i(M)/n$. 
Thus, 
\begin{align*}
\Pr&[\text{No proportional allocation exists}] \\
&\geq \Pr[\forall i \in N, d_i(1) > m/n] \\
&= \prod_{i \in N} \Pr[d_i(1) > m/n] \\
&\geq \prod_{i \in N} \left(1 - \beta \cdot \frac{m}{n}\right) \\
&\geq \prod_{i \in N} \frac{1}{1 + 2\beta \cdot \frac{m}{n}} \\
&\geq e^{-2\beta  m},
\end{align*}
where the second inequality follows from the $(\alpha, \beta)$-PDF-boundedness of $\cD$, the third inequality from our assumption $2\beta\cdot\frac{m}{n}\le 1$ and the bound $(1+x)(1-x/2) \ge 1$ which holds for all $x\in[0,1]$, and the last inequality from the bound $1+x \le e^x$ which holds for all real numbers $x$.
\end{proof}

\subsection{Existence}

Next, we establish the existence of proportional allocations (\Cref{thm:prop-ub-main}).
As discussed in \Cref{sec:overview}, we first give an algorithm for the ``small-$m$'' case, shown as \Cref{alg:matching-prop}.

\begin{algorithm}
\caption{Matching for Proportional Allocation}
\label{alg:matching-prop}
\begin{algorithmic}[1]
\Procedure{$\algpropsmall(N, M, \{d_i\}_{i\in N})$}{}
\State $E = \{(i, j) \in N \times M \mid d_i(j) \leq d_i(M)/n \text{ and } j = \argmin_{j' \in M} d_i(j')\}$
\If{a right-saturated matching exists in the graph $G = (N, M, E)$}
\State $A \gets$ the allocation corresponding to a  right-saturated matching
\State \Return $A$
\Else 
\State \Return NULL
\EndIf
\EndProcedure
\end{algorithmic}
\end{algorithm}

\begin{theorem} \label{thm:prop-small-m}
Let $\cD$ be an $(\alpha, \beta)$-PDF-bounded distribution with mean $\mu$, and let $\frac{m}{\log m} \geq \frac{80}{\alpha \mu}$ and $m \log m \leq n/40$.
There exists a polynomial-time algorithm $\algpropsmall$ that finds a proportional allocation with probability at least $1 - 1/m^4$. 
\end{theorem}

\begin{proof}
First, note that if \Cref{alg:matching-prop} does not return NULL, then the output allocation must be proportional---this is because either $d_i(A_i) = 0$ (agent $i$ does not receive any chore) or $d_i(A_i) \leq d_i(M) / n$ (agent $i$ receives one chore corresponding to an edge in $E$). 
Thus, it suffices to show that a right-saturated matching exists in $G$ with probability at least $1 - 1/m^4$.

Observe that each $i \in N$ has degree at most one, since it can be connected to only $\argmin_{j' \in M} d_i(j')$. 
This means that there exists a right-saturated matching in $G$ if and only if no right vertex is isolated (i.e., has degree zero). 
Hence, by the union bound, we have
\begin{align}
&\Pr[\text{No right-saturated matching exists in }G] \nonumber \\
&\leq \sum_{j \in M} \Pr[j \text{ is isolated in } G] \nonumber \\
&= \sum_{j \in M} \Pr[\forall i \in N, (i, j) \notin E] \nonumber \\
&= \sum_{j \in M} \prod_{i \in N} \Pr[(i, j) \notin E] \nonumber \\
&= \sum_{j \in M} \prod_{i \in N} \left(1 - \Pr[(i, j) \in E]\right). \label{eq:matching-prob-prop}
\end{align}
Let $F_{\cD}$ denote the cumulative distribution function of $\cD$, and let $\tau$ be such that $F_{\cD}(\tau) = \frac{20 \log m}{n}$. 
Note that the $(\alpha, \beta)$-PDF-boundedness of $\cD$ implies that $\tau \leq \frac{1}{\alpha} \cdot \frac{20 \log m}{n}$.

Fix $i \in N$ and $j \in M$.
We can bound $\Pr[(i, j) \in E]$ as follows:
\begin{align*}
&\Pr[(i, j) \in E] \\
&= \Pr[d_i(j) \leq d_i(M)/n  \text{ and } j = \argmin_{j' \in M} d_i(j')] \\
&\geq \Pr[d_i(j) < \tau \text{ and } \forall j' \in M \setminus \{j\}, d_i(j') \geq \tau \\
&\qquad\quad \text{ and } d_i(M \setminus \{j\}) \geq n \cdot \tau] \\
&= \Pr[d_i(j) < \tau] \cdot \prod_{j' \in M \setminus \{j\}} \Pr[d_i(j') \geq \tau] \\
&\quad \cdot \Pr[d_i(M \setminus \{j\}) \geq n \cdot \tau \mid \forall j' \in M \setminus \{j\}, d_i(j') \geq \tau] \\
&= \frac{20 \log m}{n} \cdot \left(1 - \frac{20\log m}{n}\right)^{m - 1} \\
&\quad \cdot \Pr[d_i(M \setminus \{j\}) \geq n \cdot \tau \mid \forall j' \in M \setminus \{j\}, d_i(j') \geq \tau] \\
&\geq \frac{10 \log m}{n} \\
&\quad\cdot \Pr[d_i(M \setminus \{j\}) \geq n \cdot \tau \mid \forall j' \in M \setminus \{j\}, d_i(j') \geq \tau],
\end{align*}
where the second equality follows from the independence across different items, the last equality from our choice of $\tau$, and the last inequality from Bernoulli's inequality along with our assumption that $n/40 \geq m \log m$.

Now, observe that $d_i(M \setminus \{j\}) \geq n \cdot \tau$ and $d_i(j') \geq \tau$ for $j'\in M\setminus\{j\}$ are positively correlated events. 
Applying this to the bound above, we get
\begin{align}
\Pr[(i, j) \in E] \geq \frac{10 \log m}{n}\cdot \Pr[d_i(M \setminus \{j\}) \geq n \cdot \tau]. \label{eq:edge-prop-prob}
\end{align}
Note that $d_i(M \setminus \{j\}) = \sum_{j' \in M \setminus \{j\}} d_i(j')$ is a sum of $m - 1$ independent random variables; its expectation is $(m - 1)\mu \geq m \mu / 2$. 
From our condition on $m$, we have $m\mu/2 \geq {40 \log m} \cdot 1/\alpha \geq 2n\tau$, so $(m-1)\mu/2 \ge n\tau$. 
Hence,
\begin{align}
 \Pr[d_i(M \setminus \{j\}) \geq n \cdot \tau] 
 &\geq \Pr\left[d_i(M \setminus \{j\}) \geq \frac{(m - 1)\mu}{2}\right] \nonumber \\
  &\geq  1 - \exp\left(-\frac{(m-1)\mu}{8}\right) \nonumber  \\
  &\geq  1 - \exp\left(-\frac{m\mu}{16}\right) \nonumber  \\
  &\geq  1 - \exp\left(-\frac{5\log m}{\alpha}\right) \nonumber  \\
  &\geq  1 - \exp\left(-5\log m\right) \nonumber  \\
  &=  1 - m^{-5} \nonumber  \\
 &\geq \frac{1}{2}, \label{eq:prop-sum-prob}
\end{align}
where the second inequality follows from \Cref{lem:chernoff} and the fourth inequality from our assumption that $\frac{m}{\log m} \ge \frac{80}{\alpha\mu}$.

Finally, combining \eqref{eq:matching-prob-prop}, \eqref{eq:edge-prop-prob}, and \eqref{eq:prop-sum-prob}, we get
\begin{align*}
&\Pr[\text{No right-saturated matching exists in }G] \\
&\leq m \cdot \left(1 - \frac{5 \log m}{n}\right)^n \leq m \cdot e^{-5 \log m} = \frac{1}{m^4},
\end{align*}
where the second inequality follows from the inequality $1 + x \le e^x$, which holds for all real numbers $x$.
This concludes our proof.
\end{proof}

We are now ready to prove \Cref{thm:prop-ub-main} by considering whether the value of $m$ is ``small'', ``large'', or ``medium''.

\begin{proof}[Proof of \Cref{thm:prop-ub-main}]
Suppose that the distribution $\cD$ is $(\alpha, \beta)$-PDF-bounded with mean $\mu$. 
Note that the $(\alpha, \beta)$-PDF-boundedness of $\cD$ implies that
\[
\mu = \int_{y=0}^1 f_{\cD}(y) \cdot y \diff y \geq \int_{y=0}^1 \alpha y \diff y = \frac{\alpha}{2}.
\]
We will assume throughout that $m$ is sufficiently large---more specifically, that $\frac{m}{\log m} \geq \frac{80}{\alpha \mu}$. 
Consider the following three cases.
\begin{itemize}
\item \underline{Case I}: $m \log m \leq n/40$. In this case of ``small'' $m$, we may directly apply \Cref{thm:prop-small-m}.
\item \underline{Case II}: $m \geq 2n$. In this case of ``large'' $m$, we may apply \Cref{thm:upper-main}, which allows us to find an envy-free allocation with high probability. 
Since an envy-free allocation is also proportional, we get the desired result.
\item \underline{Case III}: $m \log m > n/40$ and $m < 2n$. 
In this case of ``medium'' $m$, let $m_0$ be the largest integer such that $m_0 \log m_0 \leq n/40$. 
Note that $m_0 = \Theta\left(n/\log n\right)$. Our algorithm works as follows.
\begin{itemize}
\item Let $r = \lceil n/m_0 \rceil$.
\item Partition $M$ into $M^1, \dots, M^r$, each of size either $\lceil n/r \rceil$ or $\lfloor n/r \rfloor$.
\item Apply $\algpropsmall$ (\Cref{alg:matching-prop}) to each $M^i$ to obtain an allocation $A^i = (A^i_1, \dots, A^i_n)$.
\item Output $A = (A^1_1 \cup \cdots \cup A^r_1, \dots, A^1_n \cup \cdots \cup A^r_n)$.
\end{itemize}
Observe that if each allocation $A^i$ for $i\in\{1,\dots,r\}$ is proportional with respect to the set of chores $M^i$, then the final allocation $A$ is also proportional with respect to the set of all chores $M$. 
Now, by our choice of $r$, we have $n/r \leq m_0$; since $m_0$ is an integer, this means that $\lceil n/r\rceil \le m_0$.
It follows that $|M^i| \log |M^i| \leq n/40$. 
Furthermore, $\frac{|M^i|}{\log |M^i|} = \Theta\left(\frac{n}{\log^2 n}\right)$, which is at least $\frac{80}{\alpha \mu}$ for any sufficiently large $n$. 
Thus, we may apply \Cref{thm:prop-small-m} to each invocation of $\algpropsmall$, which implies that the probability that $A^i$ is proportional with respect to $M^i$ is at least $1 - 1/|M^i|^4 \geq 1 - O(1/m_0^4)$.
Thus, the probability that the above algorithm fails to find a proportional allocation is at most 
\begin{align*}
O\left(r \cdot \frac{1}{m_0^4}\right) 
&= O\left(\frac{n}{m_0} \cdot \frac{1}{m_0^4}\right) \\
&= O\left(\frac{\log^5 n}{n^4}\right) = o(1),
\end{align*}
as desired. \qedhere
\end{itemize}
\end{proof}

\section{Discussion}

In this work, we have derived asymptotically tight bounds on the existence of envy-free and proportional chore allocations when the disutilities are sampled from a probability distribution.
Interestingly, these bounds reveal that far fewer items are required to achieve fairness when dealing with chores in comparison to goods.
We remark that our results also have implications on two other important fairness notions: \emph{maximin share (MMS) fairness} \citep{Budish11,AzizRaSc17} and \emph{envy-freeness up to any chore (EFX)} \citep{AzizCaIg22,KobayashiMaSa25}.
Specifically, for each of these notions, an allocation satisfying it exists with high probability for any relation between $m$ and $n$.
In more detail, for MMS fairness, if $m\le n$ then every allocation that gives at most one chore to each agent is MMS-fair, while if $m > n$, since proportionality is stronger than MMS fairness, \Cref{thm:prop-ub-main} implies that an MMS-fair allocation is again likely to exist.
Similarly, for EFX, if $m\le 2n$, it is known that an EFX allocation always exists \citep{KobayashiMaSa25}; if $m > 2n$, since envy-freeness is stronger than EFX, \Cref{thm:upper-main} implies that an EFX allocation exists with high probability.

A technical question left by our work is to tighten the gap in the constant factor between \Cref{thm:upper-main,thm:lb-main}, in particular, to determine whether the factor of~$2$ in \Cref{thm:upper-main} can be improved.
One could also consider allocating a combination of goods and chores \citep{AzizCaIg22}---\citet{BenadeHaPs24} presented results on this model but only for two agents.
Finally, it would be interesting to examine whether similar results continue to hold if we relax certain assumptions on the model, as has been done in the case of goods \citep{BaiFeGo22,BaiGo22,BenadeHaPs24}.

\section*{Acknowledgments}

This work was partially supported by the Singapore Ministry of Education under grant number MOE-T2EP20221-0001 and by an NUS Start-up Grant.
We thank the anonymous reviewers for their valuable feedback.

%\newpage

\bibliographystyle{named}
\bibliography{ijcai25}

%\newpage
\appendix

\end{document}